\documentclass[preprint,authoryear]{elsarticle}
\usepackage{graphicx}
\usepackage{amsmath}
\usepackage{setspace}
\usepackage{amssymb}
\usepackage{epstopdf}
\usepackage{natbib}

\newcommand{\bA}{{\mathrm{A}}}

\newcommand{\bx}{{\mathbf{x}}}
\newcommand{\by}{{\mathbf{y}}}

\newcommand{\bM}{{\mathrm{M}}}
\newcommand{\bN}{{\mathbf{n}}}
\newcommand{\bR}{{\mathbf{r}}}

\newcommand{\bQ}{{\mathrm{Q}}}
\newcommand{\bQd}{{\mathrm{R}}}

\newcommand{\be}{\begin{equation}}
\newcommand{\ee}{\end{equation}}

\newcommand{\beq}{\begin{eqnarray}}
\newcommand{\eeq}{\end{eqnarray}}

\newcommand{\nn}{\nonumber}
\newcommand{\nd}{n'}
\newcommand{\rd}{r'}
\newenvironment{proof}{\noindent {\em Proof}\\}{\hfill $\Box$\\}

\newtheorem{theorem}{Theorem}
\newtheorem{lemma}{Lemma}

\journal{Theoretical Population Biology}

\begin{document}
\begin{frontmatter}
\title{Exact coalescent likelihoods for unlinked markers in finite-sites mutation models}

\author[db]{David Bryant\corref{cor1}}
\ead{david.bryant@otago.ac.nz}

\author[ar]{Arindam RoyChoudhury}
\ead{ar2946@columbia.edu}

\author[rb]{Remco Bouckaert}
\ead{remco@cs.auckland.ac.nz}

\author[jf]{Joseph Felsenstein}
\ead{joe@gs.washington.edu}

\author[nr]{Noah Rosenberg}
\ead{noahr@stanford.edu}

\cortext[cor1]{Corresponding author}

\address[db]{Department of Mathematics and Statistics, University of Otago, Dunedin, New Zealand, and the Allan Wilson Centre for Molecular Ecology and Evolution. ph: +64 3 4797889. fax +64 3 479 8427.}
\address[ar]{Mailman School of Public Health, Columbia University, New York, USA.}
\address[rb]{Computational Evolution Group, University of Auckland, Auckland, New Zealand}
\address[jf]{Department of Genome Sciences and Department of
Biology, University of Washington, Box 355065, Seattle, WA 98195-5065.}
\address[nr]{Department of Biology, Stanford University, Stanford, California, USA}

\begin{abstract}
We derive exact formulae for the allele frequency spectrum under the coalescent with mutation, conditioned on allele counts at some fixed time in the past. We consider unlinked biallelic markers mutating according to a finite sites, or infinite sites, model.  This work extends the coalescent theory of unlinked biallelic markers, enabling fast computations of allele frequency spectra in multiple populations. Our results have applications to demographic inference, species tree inference, and the analysis of genetic variation in closely related species more generally.
\end{abstract}


\end{frontmatter}

\section{Introduction}

A key insight of coalescent theory is that for neutral markers the process describing the genealogical relationships between individuals in the sample can be separated from the process describing the accumulation of mutations. This separation makes the coalescent a powerful tool for modelling and inference, particularly in combination with modern Monte-Carlo methods for inference in population genetics (reviewed in \cite{Felsenstein04,Marjoram06,Wakeley09}).

In this paper we bring the genealogical and mutation processes back together again. We derive analytical expressions for allele frequency spectra given by the coalescent process with finite-sites and infinite-sites models of mutation conditioned on the number of lineages, and allele frequencies, at some fixed time in the past. Our results apply to unlinked biallelic markers. 

Analytical formulae for frequency spectra in single populations have been known and widely used for some time. These have been derived for the infinitely-many-alleles model and the finite-sites model 
using diffusion approximations \citep{Ewens72,Ewens04} and for the infinite-sites model by solving an exact recurrence \citep{Tavare84,Ewens04}. 

Here we derive formulae for the frequency spectra conditioned on allele and lineage counts in the past. We assume the standard coalescent process operating on a neutral unlinked biallelic marker, with either the  finite-sites or infinite-sites model of mutation. The {\em conditional frequency spectrum} is then given by
\[\Pr[\bR_0 = r| \bN_0=n, \bN_\tau = n_\tau, \bR_\tau = r_\tau] \]
where $\bR_0$ is the number of lineages carrying the derived allele at the present, $\bN_0$ is number of lineages sampled at the present, $\bN_\tau$ is the number of distinct ancestral lineages at some time $\tau$ in the past, and $\bR_\tau$ is the number of these ancestral lineages with the derived allele. 

If we exclude the possibility of mutation then  allele frequencies change only through coalescence events.  \cite{Slatkin96} showed that the frequencies followed an urn model and  derived the a closed form expression for the conditional frequency spectrum:
\begin{equation}
\Pr[\bR_0 = r| \bN_0=n, \bN_\tau = n_\tau, \bR_\tau = r_\tau] = \frac{\binom{r-1}{r_\tau-1} \binom{n-r-1}{n_\tau-r_\tau - 1}}{\binom{n-1}{n_\tau-1}}. \label{eq:slatkin}
\end{equation}

The importance of these conditional probability formulae lies in their use for computing frequency spectra or likelihoods across multiple populations. \cite{Nielsen98a} showed how \eqref{eq:slatkin} could be used to maximum likelihood estimates of divergence times and population parameters. \cite{RoyChoudhury08} used a dynamic programming framework to compute exact likelihoods for entire population trees. 

The methods of both \cite{Nielsen98a} and \cite{RoyChoudhury08} assume models which exclude the possibility of mutation in all populations except the root population. In many contexts, assuming a zero mutation rate is completely appropriate. One important example is the analysis of single nucleotide polymorphism (SNP) data from closely related populations. However other markers have higher mutation rates, and SNP data can be used to analyse individuals from more divergent populations or even different species. In both cases, assuming a zero mutation rate is unrealistic.

Our expressions for the conditional frequency spectrum can be viewed as an extension of Slatkin's formula \eqref{eq:slatkin} to handle finite-sites or infinite-sites mutation.  Our formulae can be incorporated into the same dynamic programming framework as that developed by \cite{RoyChoudhury08}, and used to efficiently compute exact likelihoods, with mutation, across multiple populations or species. In a companion paper\footnote{Bryant, D., Bouckaert, R., Felsenstein, J., Rosenberg, N.A., RoyChoudhury, A. 2001 Inferring Species Trees Directly from Biallelic Genetic Markers: Bypassing Gene Trees in a Full Coalescent Analysis. Submitted to {\em Mol. Biol. Evol.} ArXiv preprint 0910.4193 available at www.arxiv.org} we apply these formulae within a  likelihood algorithm which is used to species trees and population parameters. The algorithm is fast enough to infer species trees and parameters for dozens of individuals using hundreds of thousands of markers.

\section{Setting the scene}

%
%


Consider $n$ individuals sampled from a Wright-Fisher population. Under the neutral  coalescent model, the number of distinct ancestral lineages at some time $t$ in the past follows a pure death process. The rate, backwards in time, for going from $k$ distinct ancestral lineages to $k-1$ lineages is $k(k-1)/2$. We assume the standard rescaling of time in terms of effective population size so that one unit of time corresponds to $2N_e$ generations.

 Consider some time $\tau$ before the present. For any $t$ such that $0 \leq t \leq \tau$, let $\bN_t$ denote the number of distinct lineages ancestral to the sample at time $t$. The conditional distribution of $\bN_t$ given $\bN_0$ was obtained by \cite{Tavare84} (see also \cite{Griffiths80}):
\beq \Pr[\bN_t = m | \bN_0=n] &=&  \sum_{k=m}^{n} e^{-k(k-1)  t/2} \frac{(2k-1)(-1)^{k-m} {m}_{(k-1)} {n}_{[k]}}{m! (k-m)! {n}_{(k)}}, \label{eq:Tavare} \eeq
where $n_{[k]} = n(n-1)(n-2)\cdots(n-k+1)$ and $n_{(k)} = n(n+1)\cdots(n+k-1)$.

\begin{figure}
\begin{center}
\includegraphics[width=3in]{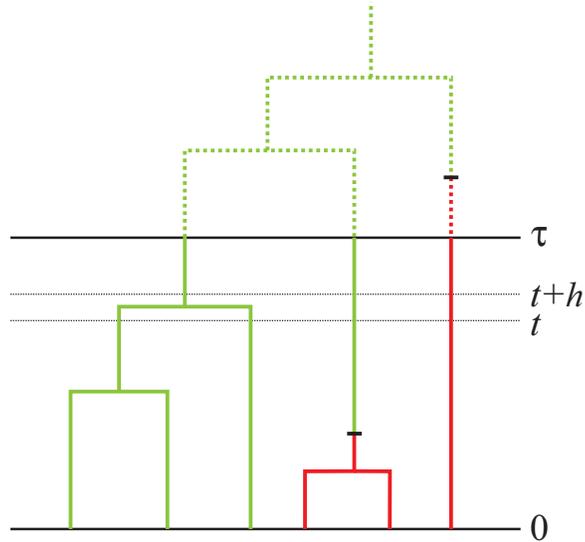}
\end{center}
\caption{\label{fig:CoalescentPic} An instance of a coalescent process followed by finite-sites mutation. In this example, $(\bN_0,\bR_0) = (6,3)$, $(\bN_t,\bR_t) = (4,1)$, $(\bN_{t+h},\bR_{t+h}) = (3,1)$ and $(\bN_\tau,\bR_\tau) = (3,1)$.}
\end{figure}

We consider a locus with two alleles, which we label {\em red} and {\em green}. For any $t$ with $0 \leq t \leq \tau$, let $\bR_t$ denote the number of lineages that carry the red allele at time $t$ before the present. Hence $0 \leq \bR_t \leq \bN_t$. For example, in Figure~\ref{fig:CoalescentPic} we have $\bN_0 = 6$, $\bR_0 = 3$, $\bN_t = 4$, $\bR_t = 1$, $\bN_\tau = 3$ and $\bR_\tau = 1$.

The neutral coalescent model with mutation follows a two-step process. First, a tree  is generated for the sample under the coalescent process. This step is followed by the mutation process, in which nodes of the tree are designated as red or green:
\begin{itemize}
\item Under the {\em infinite sites} model conditional on a polymorphic site, a single mutation is placed uniformly at random along the tree. All descendants of the mutation receive the derived {\em red} allele, whereas the remainder of the tree is assigned the ancestral green allele.
\item Under the {\em finite sites} model, mutation along a lineage is modelled by a continuous time Markov chain with rate $u$ of mutating from red to green and rate $v$ of mutating from green to red. It is assumed that the mutation process is at stationarity along the ancestral lineage, so that the root of the tree is assigned a red allele with probability $\pi = \frac{v}{u+v}$ and a green allele with probability $1-\pi = \frac{u}{u+v}$. 
\end{itemize}

\cite{Tavare84} explored a model combining the coalescent with an infinite alleles mutation model. Equation \eqref{eq:Tavare} above is a special case of a more general result of Tavar\'e involving the distribution of the number of {\em lines of descent}. A line of descent is a subtree of the gene tree whose root is one individual among the ancestral lineages and which includes all descendants in the gene tree not separated from the root by a mutation. Looking backwards in time, the number of lines of descent follows the same process as a coalescent except that lineages can be `removed' by mutation. 

\cite{Ethier87} developed a forward-time process for the coalescent with mutation for the infinite-sites model, acting on multiple linked sites. They obtained stationary probabilities  via a recursion, which can be evaluated using Monte-Carlo techniques \citep{Griffiths96,Stephens00} and  extended to multiple species \citep{Nielsen98}. An alternative recursion was developed by \cite{Fu98} and used to compute a range of related probabilities.  The frequency spectrum at a single segregating site satisfy
\be
\Pr[\bR=r|\bN=n] \propto \frac{1}{r} \label{eq:InfiniteStationary}
\ee
(see, for example, eq. 3.8 of \cite{Griffiths98} or eq. 9.63 in \cite{Ewens04}). Consequently, the expressions for the frequency spectrum of polymorphic sites do not involve the population size $\theta$  (see also \cite{RoyChoudhury10}), an observation also made by \cite{Ewens72} for the infinitely-many-alleles model.

\cite{Tian09} propose a model that combines the coalescent and mutation, but that bears only a superficial resemblance to the coalescent-with-mutation model studied here. Their `colored coalescent' makes the key simplifying assumption that after a coalescence, the allelic state of the parent is independent of the state of the children. This assumption conveniently implies that mutations in different lineages are independent, but it is clearly not an appropriate proxy for the coalescent with mutation.

\section{Conditioned frequency spectra under a finite-sites mutation model}

%
%
%

In this section we assume the finite sites model of mutation and derive an expression for $\Pr[\bR_0 = r | \bN_0 = n,\bN_\tau = n_\tau, \bR_\tau = r_\tau] $, the probability of observing $r$ red alleles in a sample of $n$ individuals taken from the present, conditional on those individuals having $n_\tau$ ancestral lineages at time $\tau$ in the past, $r_\tau$ of which also carried the red allele.

Fix $n_\tau$ and $r_\tau$ and define, for $0 \leq t \leq \tau$ and $0 \leq r_t \leq n_t$,
\begin{equation}
f_t(n_t,r_t)  =  \Pr[\bR_t = r_t | \bN_t = n_t,\bN_\tau = n_\tau, \bR_\tau = r_\tau] \Pr[\bN_\tau=n_\tau|\bN_t = n_t]. \label{eq:fdef}
\end{equation}
We abbreviate this quantity to $\Pr[\bR_t|\bN_t,\bN_\tau, \bR_\tau]\Pr[\bN_\tau| \bN_t]$. 
The goal is to determine $f_0(n_0,r_0)$, since $\Pr[\bN_\tau| \bN_t]$ can be computed using \eqref{eq:Tavare}.

\begin{lemma} \label{lem:cond}
Suppose that $t < t+h < \tau$ (see Figure~\ref{fig:CoalescentPic}). The following identities hold:
\beq 
\Pr[\bN_{t+h}|\bN_t,\bN_\tau,\bR_\tau]
& = & \frac{\Pr[\bN_\tau| \bN_{t+h}]\Pr[\bN_{t+h}|\bN_{t}] } { \Pr[\bN_\tau| \bN_{t}] }, \label{eq:cond1} \\
\Pr[\bR_{t+h}|\bN_{t},\bN_{t+h},\bN_\tau,\bR_\tau] &=& \Pr[\bR_{t+h}|\bN_{t+h},\bN_\tau,\bR_\tau], \label{eq:cond2} \\
\Pr[\bR_t|\bN_{t},\bN_{t+h},\bR_{t+h},\bN_\tau,\bR_\tau] &=& \Pr[\bR_t|\bN_{t},\bN_{t+h},\bR_{t+h}] . \label{eq:cond3}
\eeq
\end{lemma}
\begin{proof}
 The combined coalescence and mutation process is simulated by first generating a coalescent tree in the direction of increasing $t$, and then evolving the mutation process in the opposite direction. As mutation is a forward process, we have
\begin{align*}
\Pr[\bR_\tau|\bN_{t+h},\bN_\tau] &= \Pr[\bR_\tau|\bN_\tau] \\
\intertext{from which we infer}
\Pr[\bN_{t+h}|\bN_t,\bN_\tau,\bR_\tau] &=\Pr[\bN_{t+h}|\bN_t,\bN_\tau]\\
\intertext{giving \eqref{eq:cond1}. Once we condition on $\bR_\tau$, the state of $\bR_{t+h}$ depends only on the coalescent events between $t+h$ and $\tau$, and the associated mutation process. Hence}
\Pr[\bR_{t+h}|\bN_{t+h},\bN_\tau,\bR_\tau,\bN_t] & = \Pr[\bR_{t+h}|\bN_{t+h},\bN_\tau,\bR_\tau].
\intertext{The final identity \eqref{eq:cond3} follows from the mutation process being Markov, once the coalescent process is fixed.}
\end{align*}
\end{proof}
  
We marginalise over $\bN_{t+h}$ and $\bR_{t+h}$ and  apply  Lemma~\ref{lem:cond} to obtain an expression for $f_t$ in terms of $f_{t+h}$.

\begin{align}
f_t(n_t,r_t) & =  \Pr[\bR_t|\bN_t,\bN_\tau, \bR_\tau]\Pr[\bN_\tau| \bN_t] \nn \\
& =  \sum_{n_{t+h}= n_t}^{n_\tau} \sum_{r_{t+h}=0}^{r_{t+h}} \Pr[\bR_t|\bN_t, \bN_{t+h},\bR_{t+h},\bN_\tau, \bR_\tau] \Pr[\bN_{t+h},\bR_{t+h}|\bN_t,\bN_\tau, \bR_\tau]\Pr[\bN_\tau| \bN_t] \nn \\
& =  \sum_{n_{t+h}= n_t}^{n_\tau} \sum_{r_{t+h}=0}^{r_{t+h}} \Pr[\bR_t|\bN_t, \bN_{t+h},\bR_{t+h}] \Pr[\bR_{t+h}|\bN_{t+h},\bN_t,\bN_\tau, \bR_\tau] \Pr[\bN_{t+h}|\bN_t,\bN_\tau] \Pr[\bN_\tau| \bN_t] \nn  \\
& =  \sum_{n_{t+h}= n_t}^{n_\tau} \sum_{r_{t+h}=0}^{r_{t+h}} \Pr[\bR_t|\bN_t, \bN_{t+h},\bR_{t+h}] \Pr[\bR_{t+h}|\bN_{t+h},\bN_\tau, \bR_\tau] \Pr[\bN_{t+h}|\bN_t,\bN_\tau] \Pr[\bN_\tau| \bN_t] \nn \\
&= \sum_{n_{t+h}= n_t}^{n_\tau} \sum_{r_{t+h}=0}^{r_{t+h}} \Pr[\bR_t|\bN_t, \bN_{t+h},\bR_{t+h}] \Pr[\bR_{t+h}|\bN_{t+h},\bN_\tau, \bR_\tau]  \Pr[\bN_\tau| \bN_{t+h}] \Pr[\bN_{t+h}|\bN_t] \nn \\
& =   \sum_{n_{t+h}= n_t}^{n_\tau} \sum_{r_{t+h}=0}^{r_{t+h}} \Pr[\bR_t|\bN_t, \bN_{t+h},\bR_{t+h}] \Pr[\bN_{t+h}|\bN_t] f_{t+h}(n_{t+h},r_{t+h}) \label{eq:geq}.
\end{align}

We now find expressions for $\Pr[\bR_t|\bN_t, \bN_{t+h},\bR_{t+h}]$ and $\Pr[\bN_{t+h}|\bN_t]$.

From the  coalescent model for neutral loci,
\begin{equation}
\Pr[\bN_{t+h} = \nd|  \bN_{t} = n] = \begin{cases}
\binom{n}{2}  h + o(h) & \textrm{ if $\nd = n-1$;}\\ 
1 - \binom{n}{2}  h + o(h) & \textrm{ if $\nd = n$;} \\
o(h) & \textrm{ otherwise.} \end{cases}
\end{equation}
Assuming  $h$ is small, we ignore events with probability $o(h)$ and hence only  consider the cases of no coalescent events ($\nd = n$) or one coalescent event $(\nd = n-1)$ between time $t$ and $t+h$. For each of these two cases, we consider what can happen to the allele counts with and without mutation.

\begin{figure}[htbp] 
   \centering
   \includegraphics[width=4in]{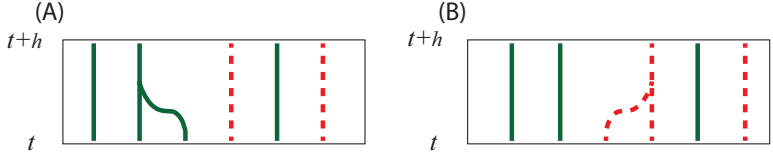} 
   \caption{Illustration of how the number of red (dashed line) or green (solid line) lineages can change at a coalescence. (A) The node at the coalescence is green, so the number of green lineages changes. (B) The coalescence is at a red node.}
   \label{fig:CoalescentOptions}
\end{figure}

First  consider the case of no coalescent events. The probability that more than one mutation occurs between time $t$ and time $t+h$ is $o(h)$, so we only consider values of $\rd$ obtained from $r$ by 0 or 1 mutation: $\rd = r-1,r+1,$ or $r$.

For $\rd = r-1$, the probability that $\bR_t = r$ is the probability that one of the $n-(r-1)$ green lineages at time $t+h$ mutated into a red lineage. Therefore
\begin{eqnarray}
\Pr[\bR_t = r | \bN_t =  \bN_{t+h} = n, \bR_{t+h} = r-1] &=&   (n - (r-1))vh + o(h). \label{eq:rn1}
\end{eqnarray}
For $\rd = r+1$, conditional on $\bN_t = \bN_{t+h}=n$, the probability that $\bR_t = r$ is the probability that one of the $r+1$ red lineages at time $t+h$ mutated into a green lineage. As a result,
\begin{eqnarray}
\Pr[\bR_t = r | \bN_t =  \bN_{t+h} = n, \bR_{t+h} = r+1] &=& (r+1)uh + o(h). \label{eq:rn2}
\end{eqnarray}
For $\rd = r$, the probability that $\bR_t = r$ is the probability that none of the $r$ red lineages or $n - r$ green lineages at time $t+h$ mutated. Thus
\begin{eqnarray}
\Pr[\bR_t = r | \bN_t =  \bN_{t+h} = n, \bR_{t+h} = r] &=&  1 - (n-r)vh - ruh + o(h). \label{eq:rn3}
\end{eqnarray}

Now consider the second case, in which there is one coalescent event and $\bN_{t+h}=n-1$. The probability that both a coalescent event and a mutation event occur between time $t$ and $t+h$ is $o(h)$, so we ignore this possibility. The number of red lineages can change, depending on whether the node at the coalescent event is red or green. If it is red (Figure~\ref{fig:CoalescentOptions} (B)) the number of red lineages will increase from time $t+h$ to time $t$. If it is green then the number of red lineages will remain the same. Hence
\begin{eqnarray}
\Pr[\bR_t = r | \bN_t = n, \bN_{t+h} = n-1, \bR_{t+h} = r-1] & = & \frac{r-1}{n-1} + o(1), \label{eq:rn4} \\
\Pr[\bR_t = r | \bN_t = n, \bN_{t+h} = n-1, \bR_{t+h} = r] &=& \frac{n-1-r}{n-1} + o(1). \label{eq:rn5} 
\end{eqnarray}

 We now substitute \eqref{eq:rn1}---\eqref{eq:rn5} into \eqref{eq:geq},  collecting products of quantities that are $o(h)$. This gives
 \begin{align}
f_t(n,r) & = \left( (n - (r-1))vh + o(h) \right) \left( 1 - \binom{n}{2}h + o(h) \right) f_{t+h}(n,r-1) \nn \\
& \quad + \left( (r+1)uh + o(h) \right) \left(1 - \binom{n}{2}h + o(h) \right) f_{t+h}(n,r+1) \nn \\
& \quad + \left(1 - (n-r)vh - ruh + o(h) \right) \left(1 - \binom{n}{2}h + o(h) \right) f_{t+h}(n,r) \nn \\
& \quad + \left(\frac{r-1}{n-1} + o(1) \right) \left( \binom{n}{2}h + o(h) \right) f_{t+h}(n-1,r-1) \nn \\
& \quad + \left(\frac{n-1-r}{n-1} + o(1) \right) \left( \binom{n}{2}h + o(h) \right) f_{t+h}(n-1,r) \nn \\
& \quad + o(h) \\
 & =  f_{t+h}(n,r-1) (n-r+1)vh \nn \\
 &\quad  + \quad  f_{t+h}(n,r+1) (r+1)uh \nn \\
 &\quad + \quad f_{t+h}(n,r) \left(1 - \binom{n}{2}  h - (n-r)vh -ruh \right) \nn \\
 & \quad + \quad f_{t+h}(n-1,r) \frac{n-1-r}{n-1} \binom{n}{2}  h \nn \\
  & \quad+\quad  f_{t+h}(n-1,r-1) \frac{r-1}{n-1} \binom{n}{2}  h \quad + \quad o(h).
\end{align}

Rearranging, dividing by $h$, and taking the limit as $h \rightarrow 0$ we obtain
\begin{align}
\frac{d}{dt} f_t(n,r) & =  -f_{t}(n,r-1) (n-r+1)v- f_{t}(n,r+1) (r+1)u \nn \\
& \quad -  f_{t}(n-1,r) \frac{n-1-r}{n-1} \binom{n}{2}  -  f_{t}(n-1,r-1) \frac{r-1}{n-1} \binom{n}{2}   \nn \\
& \quad +  f_{t}(n,r) \left(\binom{n}{2}   + (n-r)v + ru \right) \nn \\
& =  -\sum_{m=1}^n \sum_{q=0}^m \bQ_{(n,r);(m,q)}  f_t(m,q).  \label{eq:gDE} 
\end{align}
Here, $\bQ$ is a matrix with rows and columns indexed by pairs $(n,r)$, with 
\begin{align}
\bQ_{(n,r);(n,r-1)} & =  (n-r+1)v & 0<r\leq n\nn \\
\bQ_{(n,r);(n,r+1)} & =  (r+1)u & 0 \leq r < n \nn \\
\bQ_{(n,r);(n,r)} & =  - \binom{n}{2}   - (n-r)v - ru  & 0 \leq r \leq n   \label{Qdef} 
 \\
\bQ_{(n,r);(n-1,r)} & =   \frac{(n-1-r)n}{2}    & 0 \leq r < n   \nn  \\
\bQ_{(n,r);(n-1,r-1)} & =  \frac{(r-1)n}{2}   &0 <r\leq n \nn
 \end{align}
and all other entries zero.  

Note that $f_t(n,r)$ is bounded for all $t<\tau$, but might be undefined at $t = \tau$ if $n \neq n_\tau$. When $n \neq n_\tau$ we  have $\Pr[\bN_\tau = n_\tau|\bN_t=n] \rightarrow 0$ as $t \rightarrow \tau$, so $f_t(n,r) \rightarrow 0$. Hence
\begin{align}
\lim_{t \rightarrow \tau} f_t(n,r) &=  \begin{cases} 1 & \textrm{ if $n=n_\tau$ and $r = r_\tau$;} \\
0 & \textrm{ otherwise.} 
\end{cases} \\
\intertext{This result provides the boundary conditions for the differential equation \eqref{eq:gDE}. Solving and substituting $t=0,$ we obtain}
f_0(n,r) & =  \exp(\bQ \tau)_{(n,r);(n_\tau,r_\tau) } .
\end{align}

We have now established
\begin{theorem} \label{thm:transitions}
 Let $\bQ$ be the matrix defined in \eqref{Qdef}. Under the finite sites model,
\begin{equation}
\Pr[\bR_0 = r| \bN_0=n, \bN_\tau = n_\tau, \bR_\tau = r_\tau] =  \frac{\exp(\bQ \tau)_{(n,r);(n_\tau,r_\tau) } }{\Pr[\bN_\tau \!=\! n_\tau | \bN_0 \!=\! n]}.
 \end{equation}
\end{theorem}

As a corollary  to Theorem~\ref{thm:transitions}, we derive a new formula for the stationary probabilities
$\Pr[\bR_0 | \bN_0]$, which represent the probability of observing $\bR_0$ individuals carrying the red allele in a sample of size $\bN_0$. A closed form approximation for $\Pr[\bR_0 | \bN_0]$ can  be obtained by way of a diffusion approximation. Under the diffusion model, the allele proportions in the entire population have an approximately beta-distribution and so
$\Pr[\bR_0 | \bN_0]$ follows a {\em beta-binomial} distribution \citep{Ewens04}. The distribution we derive here does not use an explicit diffusion approximation but is instead based solely on the assumptions of the coalescent model. It gives probabilities very close to a beta-binomial, though the probabilities are not {\em exactly} the same. The difference presumably stems from the slightly different model assumptions underlying the standard diffusion and coalescent models.

Consider a single population from which we have taken $\bN_0 = n$ samples. 
\begin{eqnarray}
\Pr[\bR_0 = r | \bN_0 = n] & = & \sum_{n_t = 1}^{n} \sum_{ r_t = 0}^{n_t} \Pr[\bN_t = n_t, \bR_t = r_t|\bN_0 = n ]  \Pr[\bR_0 = r | \bN_0 = n, \bN_t = n_t, \bR_t = r_t] \nn  \\
& = & \sum_{n_t = 1}^{n} \sum_{ r_t = 0}^{n_t}\frac{\Pr[\bN_t = n_t, \bR_t = r_t|\bN_0 = n ] }{ \Pr[\bN_t = n_t | \bN_0 = n]}  \exp(\bQ t)_{(n,r);(n_t,r_t)} \\
& = & \sum_{n_t = 1}^{n} \sum_{ r_t = 0}^{n_t} \Pr[ \bR_t = r_t|\bN_0 = n, \bN_t = n_t ]   \exp(\bQ t)_{(n,r);(n_t,r_t)}  \\
& = & \sum_{n_t = 1}^{n} \sum_{ r_t = 0}^{n_t} \Pr[ \bR_t = r_t|\bN_t = n_t ]   \exp(\bQ t)_{(n,r);(n_t,r_t)} .
\label{eq:stat1}
\end{eqnarray}
We take the limit of the right hand side as $t \rightarrow \infty$. For this computation we examine the spectrum of $\bQ$. The structure of $\bQ$ makes this quite straight-forward.

\begin{figure}[htbp] 
   \centering
   \includegraphics[width=2in]{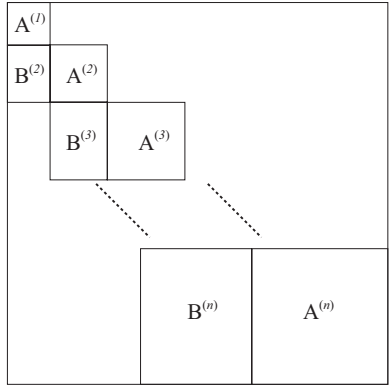} 
   \caption{The bidiagonal block structure of the matrix $\bQ$. }
   \label{fig:Qmatrix}
\end{figure}

Order the rows and columns of $\bQ$ according to the pairs $(n,r)$, in order $(1,0)$, $(1,1)$, $(2,0)$, $(2,1)$, $(2,2)$, $(3,0),\ldots$. Then $\bQ$ has block bidiagonal form as in Figure~\ref{fig:Qmatrix}. For example, when $n=4$ we have
\begin{small}
\[\bQ = \left( \begin{matrix}
 -v & u & 0 & 0 & 0 & 0 & 0 & 0 & 0\\
v & -u & 0 & 0 & 0& 0 & 0 & 0 & 0 \\
1 & 0 & -1-2v & u & 0& 0 & 0 & 0 & 0 \\
0 & 0 & 2v & -1-u-v & 2u & 0 & 0 & 0 & 0\\
0 & 1 & 0 & v & -1-2u & 0 & 0 & 0 & 0 \\
0 & 0 & 3  & 0 & 0 			& -3 - 3v & u & 0 & 0 \\
0 & 0 & 0  & \mbox{$\frac{3}{2}$} & 0 	& 3v & -3 - u-2v & 2u & 0 \\
0 & 0 & 0  & \mbox{$\frac{3}{2}$} & 0	 & 0 & 2v & -3-2u-v & 3u \\
0 & 0 & 0 & 0 & 3 			& 0 & 0 & v & -3 -3u 
\end{matrix} \right) \]
\end{small}

Since $\bQ$ is block triangular, the eigenvalues of $\bQ$ are exactly the eigenvalues of the diagonal blocks $\bA^{(1)},\ldots,\bA^{(n)}$, as can be seen by decomposing the characteristic polynomial $\det(\bQ - \lambda \mathbf{I})$ into $\det(\bA^{(1)} - \lambda \mathbf{I}) \cdot \cdots \cdot \det(\bA^{(n)} - \lambda \mathbf{I})$. Furthermore, each block $\bA^{(i)}$ equals the rate matrix of a birth death process, with the value $ i(i-1)/2 $ subtracted from the diagonal. Hence $\bA^{(i)}$ has strictly negative eigenvalues when $i>1$ (see e.g. \cite{Grimmett01}).  By inspection, $\bA^{(1)}$ has one zero eigenvalue and one negative eigenvalue. Hence $\bQ$ has one zero eigenvalue while the remaining eigenvalues are strictly negative. 

From the diagonalisation of $\bQ$ we see that
\begin{align}
\lim_{t \rightarrow \infty} \exp(\bQ t) &= \bx \by^T, \label{eq:Qxy} \end{align}
where $\by$ is a left 0-eigenvector of $\bQ$ and $\bx^T$ is a right 0-eigenvector of $\bQ$ such that $\by^T \bx = 1$. A left eigenvector is given by 
\[\by = [1,1,0,0,\ldots,0]^T.\]
The corresponding right eigenvector is found by solving the recurrence implied by $\bQ \bx = \mathbf{0}$, that is
\begin{align}
\bx^{(1)} & = \left[\frac{u}{u+v},\frac{v}{u+v}\right]^T \\
\bx^{(i)} & = \left(A^{(i)}\right)^{-1} B^{(i)} \bx^{(i-1)} \hspace{1cm} i=2,3,\ldots,n.
\end{align}
Substituting into \eqref{eq:Qxy} we have
\begin{equation}  \lim_{t \rightarrow \infty} \exp(\bQ t)_{(n,r);(n_t,r_t)} = 
\begin{cases} \mathbf{x}_{(n,r)} & \textrm{ if $n_t = 1$} \\ 0 & \textrm{ otherwise.} \end{cases} \label{eq:stat3} \end{equation}
Most of the summation terms in  \eqref{eq:stat1} are zero, leaving \begin{align}
\Pr[\bR_0 = r | \bN_0 = n]  &= \lim_{t \rightarrow \infty} (\Pr[ \bR_t = 0|\bN_t = 1 ] + \Pr[ \bR_t = 1|\bN_t = 1 ] )   \mathbf{x}_{(n,r)}  \\
& =  \mathbf{x}_{(n,r)}.
\end{align}
We now have an exact formula for allele frequency spectrum for a biallelic marker under the coalescent. Since each matrix $A^{(i)}$ is tridiagonal, each equation
\[\bx^{(i)}  = \left(A^{(i)}\right)^{-1} B^{(i)} \bx^{(i-1)} \hspace{1cm} i=2,3,\ldots,n.\]
takes only $O(n)$ time to compute \citep{Golub96}, making $O(n^2)$ time to compute all of the probabilities $\Pr[\bR_0 = r | \bN_0 = n]$. 

\begin{theorem} \label{thm:root}
Let $\bN$ and $\bR$ be the number of lineages and the number of red lineages sampled from a single population of constant size. Let $\bQ$ be the matrix defined in \eqref{Qdef} and let $\mathbf{x}$ be a non-zero solution for $\bQ \mathbf{x} = \mathbf{0}$, scaled so that $\bx_{(1,0)} + \bx_{(1,1)} = 1$. Then for all $n,r$, $\Pr[\bR \!=\! r | \bN \!=\! n] =   \mathbf{x}_{(n,r)}.$
\end{theorem}

\section{Conditional frequency spectra under the infinite-sites mutation model}

We now consider the infinite-sites mutation model. Suppose that mutations accumulate along lineages at a constant rate $\mu$ and that, at time $\tau$ in the past, all lineages carry the ancestral {\em green} allele. The allele distribution for a site under the infinite-sites model is then obtained by conditioning on there being {\em exactly} one mutation between time $0$ and time $\tau$. To this end, let $\bM_t$ denote the event that at most one mutation has occurred, over all the lineages, between time $\tau$ and time $t$.  Apart from the inclusion of this additional event and some simplifications there is little difference between the analysis here for the infinite sites model and the earlier finite sites model analysis.

Fix $n_\tau$ and $r_\tau$ and define, for $0 \leq t \leq \tau$,
\begin{equation}
g_t(n,r)  =  \Pr[\bM_t, \bR_t = r | \bN_t = n,\bN_\tau = n_\tau, \bR_\tau = r_\tau] \Pr[\bN_\tau = n_\tau | \bN_t = n]. \label{eq:fdef2}
\end{equation}
This is simply  the function $f_t(n,r)$ with the inclusion of the event $\bM_t$. We wish to determine $g_0(n,r)$. The first step is to incorporate $\bM_t$ into two of the identities in Lemma~\ref{lem:cond}.

\begin{lemma}
Suppose that $t < t+h < \tau$. Then the following identities hold:
\begin{align}
\Pr[\bM_{t+h},\bR_{t+h}|\bN_{t},\bN_{t+h},\bN_\tau,\bR_\tau] &= \Pr[\bM_{t+h},\bR_{t+h}|\bN_{t+h},\bN_\tau,\bR_\tau], \label{eq:cond2i}\\
\Pr[\bM_t,\bR_t|\bN_{t},\bM_{t+h},\bR_{t+h},\bN_{t+h},\bR_\tau,\bN_\tau] &= \Pr[\bM_t,\bR_t|\bN_{t},\bM_{t+h},\bR_{t+h},\bN_{t+h},\bR_\tau,\bN_\tau] . \label{eq:cond3i}
\end{align}
\end{lemma}
\begin{proof}
 The combined coalescence and mutation process is simulated by first generating a coalescent tree in the direction of increasing $t$, and then evolving the mutation process in the opposite direction. Once we condition on $\bR_\tau$, the state of $\bM_{t+h}$ and $\bR_{t+h}$ depends only on the coalescent events between $t+h$ and $\tau$, and the mutation process. Hence
\begin{align*}
\Pr[\bM_{t+h},\bR_{t+h}|\bN_{t+h},\bN_\tau,\bR_\tau,\bN_t] & = \Pr[\bM_{t+h},\bR_{t+h}|\bN_{t+h},\bN_\tau,\bR_\tau].
\intertext{Note that $\bM_\tau$ holds by assumption. As in Lemma~\ref{lem:cond}, the final identity \eqref{eq:cond3i} follows from the mutation process being Markov once the coalescent process is fixed.}
\end{align*}
\end{proof}

We marginalise over $\bN_{t+h}$ and $\bR_{t+h}$  and apply \eqref{eq:cond2i} and \eqref{eq:cond3i} to give an expression for $g_t$ in terms of $g_{t+h}$. This differs from the finite sites case only by the inclusion of the events $\bM_t$.
\begin{align}
g_t(n_t,r_t) & =  \Pr[M_t, \bR_t|\bN_t,\bN_\tau, \bR_\tau]\Pr[\bN_\tau| \bN_t] \nn \\
& =  \sum_{n_{t+h}=1}^{n_t} \sum_{r_{t+h}=0}^{n_{t+h}} \Pr[M_t,\bR_t|\bN_t, M_{t+h},\bN_{t+h},\bR_{t+h},\bN_\tau, \bR_\tau] \Pr[M_{t+h},\bN_{t+h},\bR_{t+h}|\bN_t,\bN_\tau, \bR_\tau]\Pr[\bN_\tau| \bN_t] \nn \\
& =  \sum_{n_{t+h}=1}^{n_t} \sum_{r_{t+h}=0}^{n_{t+h}} \Pr[M_t,\bR_t|\bN_t, M_{t+h},\bN_{t+h},\bR_{t+h}] \Pr[M_{t+h},\bR_{t+h}|\bN_{t+h},\bN_\tau, \bR_\tau] \Pr[\bN_{t+h}|\bN_t,\bN_\tau] \Pr[\bN_\tau| \bN_t] \nn  \\
& =  \sum_{n_{t+h}=1}^{n_t} \sum_{r_{t+h}=0}^{n_{t+h}} \Pr[M_t,\bR_t|\bN_t, M_{t+h},\bN_{t+h},\bR_{t+h}] \Pr[M_{t+h},\bR_{t+h}|\bN_{t+h},\bN_\tau, \bR_\tau]\Pr[\bN_\tau| \bN_{t+h}] \Pr[\bN_{t+h}|\bN_t]  \nn  \\
& =   \sum_{n_{t+h}=1}^{n_t} \sum_{r_{t+h}=0}^{n_{t+h}} \Pr[M_t,\bR_t|\bN_t, M_{t+h},\bN_{t+h},\bR_{t+h}] \Pr[\bN_{t+h}|\bN_t] g_{t+h}(n_{t+h},r_{t+h}) \label{eq:geqi}.
\end{align}

As before, many terms in the final summation of \eqref{eq:geqi} are $o(h)$. We consider the    exceptions. 


 \begin{align*} \intertext{ (i) a coalescent event, where the parent is red }
\Pr[\bM_t,\bR_t = r | \bN_t = n, \bM_{t+h}, \bN_{t+h} = n-1, \bR_{t+h} = r-1] & =  \frac{r-1}{n-1} + o(1); \nn 
\intertext{(ii) a coalescent event where the parent is green}
\Pr[\bM_t,\bR_t = r | \bN_t = n, \bM_{t+h}, \bN_{t+h} = n-1, \bR_{t+h} = r] &= \frac{n-1-r}{n-1} + o(1);  \nn 
\intertext{(iii) the case of the mutation ancestral to all red lineages ($r=1$)}
\Pr[\bM_t,\bR_t = 1|\bN_t=n,\bM_{t+h},\bR_{t+h}=0,\bN_{t+h}=n] &= n\mu h + o(h);\nn
\intertext{(iv)  the case of no coalescent events and no mutation events}
\Pr[\bM_t,\bR_t=r|\bN_t=n,\bM_{t+h},\bR_{t+h}=r,\bN_{t+h}=n] &= 1 - n \mu h + o(h) \nn 
\end{align*}
The main difference with the finite sites model  is that having $\bM_t$ hold implies at most one mutation event occurring between time $t$ and $\tau$. This event is handled in case (iii).

Substitute cases (i)---(iv)  into \eqref{eq:geqi} to obtain
\begin{align*}
g_t(n,r) & =  \left(\frac{r-1}{n-1} + o(1)\right)\left(\binom{n}{2}h + o(h)\right) g_{t+h}(n-1,r-1) \\ 
& \qquad + \left(\frac{n-1-r}{n-1} + o(1)\right)\left(\binom{n}{2}h + o(h)\right) g_{t+h}(n-1,r)\\ 
& \qquad + \delta(r,1)  \left(nuh + o(h) \right)  g_{t+h}(n,0) \\
& \qquad + \left(1 - nuh + o(h)\right) \left(1 - \binom{n}{2}h + o(h)\right) g_{t+h}(n,r),
\end{align*}
where $g_t(n,r)$ and $g_{t+h}(n,r)$ are zero unless $0 \leq r \leq n$ and $1 \leq n$.

Rearranging and taking the limit as $h \longrightarrow 0$ we obtain 
\begin{align}
\frac{d}{dt} g_t(n,r) & = - \sum_{\nd,\rd} \bQd_{(n,r);(n'r')} g_t(\nd,\rd), 
\end{align}
where $\bQd$ is the matrix given by
\begin{align}
\bQd_{(n,0);(n,1)} & =  nu && 0  < n \nn \\
\bQd_{(n,r);(n-1,r)} & =   \frac{(n-1-r)n}{2}    && 0 \leq r < n     \label{Qdefb}  \\
\bQd_{(n,r);(n-1,r-1)} & =  \frac{(r-1)n}{2}   &&0 <r\leq n \nn \\
\bQd_{(n,r);(n,r)} & =  - \binom{n}{2}   - nu \nn && 0 \leq r \leq n
\end{align}
and all other entries $0$.

As in the finite sites case we have the boundary conditions
\begin{align}
\lim_{t \rightarrow \tau} g_t(n,r) &=  \begin{cases} 1 & \textrm{ if $n=n_\tau$ and $r = r_\tau$;} \\
0 & \textrm{ otherwise.} \end{cases} \\
\intertext{Solving for $\tau=0$ gives}
g_0(n,r) & =  \exp(\bQd \tau)_{(n,r);(n_\tau,r_\tau) } 
\end{align}
and so 
\begin{align}
\Pr[\bM_0,\bR_0 = r| \bN_0=n, \bN_\tau = n_\tau, \bR_\tau = r_\tau] & = f_0(n,r) \\
& =  \frac{\exp(\bQd \tau)_{(n,r);(n_\tau,r_\tau) } }{\Pr[\bN_\tau = n_\tau | \bN_0 = n_0]}.
\end{align}

Normalising now gives 
\begin{theorem}
 Let $\bQd$ be the matrix defined in \eqref{Qdefb}. Under the infinite-sites model, at a segregating site
\begin{equation}
\Pr[\bR_0 = r| \bN_0=n, \bN_\tau = n_\tau, \bR_\tau =0] =  \frac{\exp(\bQd \tau)_{(n,r);(n_\tau,0) } }{\sum_{\rd=1}^n \exp(\bQd \tau)_{(n,\rd);(n_\tau,0) } }.
 \end{equation}
for $0<r<n$.
\end{theorem}

We note that if $r_\tau>0$ then the distribution of $\bR_0$ is determined solely by the coalescent events, and
\begin{equation}
\Pr[\bR_0 = r| \bN_0=n, \bN_\tau = n_\tau, \bR_\tau = r_\tau] = \frac{\binom{r-1}{r_\tau-1} \binom{n-r-1}{n_\tau - r_\tau - 1}}{\binom{n-1}{n_\tau-1}},
\end{equation}
see \cite{Slatkin96}. 


\section{Conclusion}

We have derived exact formula for coalescent-based likelihoods from unlinked biallelic markers. Our results generalise urn model results of \cite{Slatkin96} by incorporating mutation, under both a finite sites or infinite sites model. As a consequence, the methods developed by \cite{Nielsen98a} and \cite{RoyChoudhury08} for analysing SNP data from single and multiple populations can also be extended to include mutation. This will have less of an impact for single population analyses, but will be of increasing importance in analyses of multiple, closely related, species. 

One consequence of our results which will have immense practical significance is that coalescent likelihood calculations can be carried out without the need to integrate over gene trees. This  is critical if coalescent analyses are to go beyond the analyses of small numbers of genes. Indeed, the formulae derived here enable a full coalescent analysis of hundreds of thousands of unlinked SNP loci.

\subsection*{Acknowledgements}

We sincerely thank Elizabeth A. Thompson for her valuable insight and scholarly support throughout this project. Her help was instrumental for the execution of this project. DB was funded by a Marsden grant, an Alexander von Humboldt Fellowship (kindly hosted by William Martin) and the Allan Wilson Centre for Molecular Ecology and Evolution. AR was supported in part by NIH program project grant GM-45344, and by NIH Grant R01 GM071639-01A1 (PI: J Felsenstein). JF was funded by National Institutes of
Health grant R01 GM071639, by National Institutes of Health grant
R01 HG004839 (PI: Mary Kuhner), and by interim ``life support'' funds from the
Department of Genome Sciences, University of Washington. NAR was funded by NSF grant DBI-1062394.

\bibliographystyle{elsarticle-harv}

\end{document}